\documentclass[12pt]{article}
\usepackage{amssymb,amsmath,amsthm,mathrsfs,latexsym,amsxtra,graphicx,appendix,epstopdf,feynmf,subfig,hyperref,setspace,fix-cm,color,cite,hyperref}

\usepackage{amsfonts}
\usepackage{mathtools}

\numberwithin{equation}{section}
\onehalfspace
\begin{document}

\newcommand{\Vast}{\bBigg@{5}}
\newtheorem*{theorem}{Theorem}
\newtheorem{conjecture}{Conjecture}[section]
\newtheorem{corollary}{Corollary}
\newtheorem{lemma}{Lemma}
\newtheorem*{remark}{Remark}
\newtheorem{claim}{Claim}
\newtheorem{prop}{Proposition}
\theoremstyle{definition}
\newtheorem{definition}{Definition}[section]
\renewcommand\qedsymbol{$\blacksquare$}

\newcommand{\be}{\begin{equation}}
\newcommand{\ee}{\end{equation}}
\newcommand{\bea}{\begin{eqnarray}}
\newcommand{\eea}{\end{eqnarray}}
\newcommand{\bb}{\mathbb}
\newcommand{\mrm}{\mathrm}
\newcommand{\scr}{\mathscr}
\newcommand{\p}{\partial}
\def\e{{\rm e}}
\def\bz{{\bar z}}
\def\bw{{\bar w}}
\def\p{{$|\Phi\rangle$}}
\def\pp{{$|\Phi^\prime\rangle$}}
\def\cO{{\mathcal O}}
\def\cH{\mathcal{H}}
\def\cF{\mathcal{F}}
\def\cL{\mathcal{L}}

\newcommand{\CC}{\mathbb{C}} 
\newcommand{\ZZ}{\mathbb{Z}} 
\newcommand{\NN}{\mathbb{N}} 

\newcommand{\T}{R}

\newcommand{\lcm}{\mathrm{lcm}}

\newcommand{\ie}{{\it i.e.~}}
\def\eg{{\it e.g.~}}

\newcommand{\comment}[1]{{\bf\color{blue}[#1]}}

\begin{titlepage}
\begin{center}

\hfill \\
\hfill \\
\vskip 0.75in

{\Large \bf The Spectrum of Permutation Orbifolds}\\

\vskip 0.4in

{\large Christoph A.~Keller${}^{a}$ and Beatrix J. M\"uhlmann$^{b}$
}\\
\vskip 4mm

${}^{a}$
{\it Department of Mathematics, ETH Zurich, CH-8092 Zurich, Switzerland} \vskip 1mm
${}^{b}$
{\it Department of Physics, ETH Zurich, CH-8092 Zurich, Switzerland} \vskip 1mm

\texttt{christoph.keller@math.ethz.ch, b.muhlmann@uva.nl}

\end{center}

\vskip 0.35in

\begin{center} {\bf ABSTRACT } \end{center}
We study the spectrum of permutation orbifolds of 2d CFTs. We find examples where the light spectrum grows faster than Hagedorn, which is different from known cases such as symmetric orbifolds.
We also describe how to compute their partition functions using a generalization of Hecke operators.

\vfill

\noindent \today

\end{titlepage}

\setcounter{tocdepth}{2}

\section{Introduction}

In the context of the $AdS_3/CFT_2$ correspondence, one is interested in families of 2d CFTs with a large central charge limit. To construct explicit examples of such families is surprisingly hard, since generically the number of light states of a family will diverge in the large central charge limit. The best known example with finite spectrum are symmetric orbifold theories \cite{Strominger:1996sh}, that is CFTs whose tensor product one orbifolds by the symmetric group $S$. In that case the growth of light states is given by \cite{Keller:2011xi,Hartman:2014oaa}
\be\label{symgrowth}
\rho_S(\Delta) \approx e^{2\pi \Delta}\ .
\ee
The exponential growth indicates that we are in a stringy regime with Hagedorn growth. From holography one expects that there should be many examples with supergravity growth $\rho(\Delta)\approx e^{\sqrt{a\Delta}}$. However, no such CFTs have been constructed explicitly. 
The goal of this note is to find theories with growth behavior different from (\ref{symgrowth}), which can then be interpreted as describing different physics. In particular we find an example whose growth is super-Hagedorn, reminiscent of the entropy of black holes in flat space for instance. To our knowledge this is the first explicit such example.

To achieve this we consider permutation orbifolds \cite{Haehl:2014yla,Belin:2014fna,Belin:2015hwa}:
We start out with a modular invariant partition function of a seed theory
\be
 Z(\tau) = \sum_{\Delta\in \ZZ_{\geq0}} \rho(\Delta) q^{\Delta-c/24}\ ,\qquad q= e^{2\pi i\tau}\ ,
\ee
with $\rho(\Delta)\in \ZZ_{\geq0}$.
For simplicity we are assuming here a holomorphic theory, but it is straightforward to generalize our results to non-holomorphic theories.
Let $G_N < S_N$ be a subgroup of the symmetric group acting on the set $\{1,\ldots,N\}$ in the standard way. 
Our starting point is then the expression for the partition function of permutation orbifolds given in
\cite{Bantay:1997ek},
\be\label{Bantay}
Z_{G_N}(\tau) = \frac{1}{|G_N|}\sum_{hg=gh} Z_{(h,g)}(\tau)
\ee
The sum here is over all $g,h\in G_N$ which commute. We can think of the sum over $g$ as labelling the twisted sector states and the sum over $h$ as projecting onto the $G_N$ invariant states in a given twisted sector.
The functions $Z_{(h,g)}$ are given by the following prescription:
A pair of commuting elements $g,h$ generate an Abelian subgroup of $S_N$,
which of course acts on the set $\{1,\ldots,N\}$ by permutation
of the elements.
We denote by $O(h,g)$ the set of orbits of this action. For each orbit $\xi \in O(h,g)$
we define the modified modulus $\tau_\xi$ as follows:
First, let $\lambda_\xi$ be the size of the $g$
orbit in $\xi$, and $\mu_\xi$ the number of $g$ orbits in $\xi$,
so that $\lambda_\xi \mu_\xi =|\xi|$.
Let $\kappa_\xi$ be the smallest non-negative integer
such that $h^{\mu_\xi} g^{-\kappa_\xi}$ is in
the stabilizer of $\xi$.
Then 
\be\label{Ztauchi}
Z_{(h,g)}(\tau)= \prod_{\xi\in O(h,g)}Z(\tau_\xi) \qquad \textrm{with}\quad
\tau_\xi = \frac{\mu_\xi\tau + \kappa_\xi}{\lambda_\xi}\ .
\ee
As we are interested in the limit $N\to\infty$, it is often more useful to shift the partition function such that the leading term is 1 rather than $q^{-cN/24}$, giving $\tilde Z_{G_N}(\tau):=Z_{G_N}(\tau)q^{cN/24}$.

The \emph{untwisted sector} partition function $Z^u_{G_N}(\tau)$ is given by the sum over terms with $g=1$. We will use the fact that its coefficients $\rho^u_{G_N}(\Delta)$ give a lower bound for $\rho_{G_N}(\Delta)$.
It is straightforward to check that $Z^u_{G_N}$ can be expressed in terms of the cycle index of $G_N$ as 
\be
Z^u_{G_N}=\chi(G_N; Z(\tau),\dots, Z(N\tau))\ .
\ee
We will consider the limit $N\to\infty$. For general families $G_N$ the $\rho_{G_N}$ will not converge. A necessary and sufficient condition for convergence is that $G_N$ be \emph{oligomorphic} \cite{MR1066691,Haehl:2014yla,Belin:2015hwa}.
We will consider two such families, or more precisely, two types of action \cite{MR2378039}: The direct product action $S_{\sqrt{N}}\times S_{\sqrt{N}}$ and the wreath product action $S_{\sqrt{N}}\wr S_{\sqrt{N}}$. We can of course iterate this $d$ times. Note that for convenience we will choose the order of the groups in such a way that their iterated products always act on $N$ elements. 

Our main result is that we give a lower bound for the number of states of such orbifold theories in the limit $N\to\infty$. For a function $g(z)$ we denote by $g^d(z):= \underbrace{g\circ g\circ \cdots \circ g}_d(z)$, that is $g$ iterated $d$ times. Similarly we denote by $S_\times^d$ and $S^d_\wr$ the iterated direct and wreath product actions respectively of the infinite permutation groups $S$. 
We define $a_n \approx b_n$ to mean $\lim_{n\to\infty} \frac{\log a_n}{\log b_n}=1$, and similar for $a_n \gtrsim b_n$.
For the wreath product we find 
\begin{prop}\label{wreath}
\be
\rho_{S^d_\wr}(\Delta) \gtrsim e^{b\Delta/ \log^d(b\Delta) }
\ee
where $b$ is a positive constant given by $b= \pi^2c/6$, with $c$ the central charge of the seed theory.
\end{prop}
For the direct product action we find super-Hagedorn growth:
\begin{prop}\label{direct}
\be
\rho_{S^d_\times}(\Delta) \gtrsim e^{(d-1)\frac{\Delta}{\Delta_1}{\log (\Delta/\Delta_1)}}
\ee
where $\Delta_1$ is a positive constant given by the weight of the lightest state in the theory.
\end{prop}
These propositions are proven in section~\ref{ss:ut}.

\section{Combinatorics}

\subsection{Orbits}
Given two permutation groups $G_1, G_2$ acting on $X_1,X_2$ respectively, we can define the \emph{direct action} of the direct product $G_1 \times G_2$ on $X_1 \times X_2$ by $(g_1,g_2)\cdot(x_1,x_2)=(g_1\cdot x_1,g_2\cdot x_2)$ and the \emph{imprimitive action} of the wreath product $G_1\wr G_2$ on $X_1\times X_2$ by $(f(x_2)\cdot x_1,x_2)$ for $f\in G_1^{X_2}$ and $(x_1,g\cdot x_2)$ for $g\in G_2$ \cite{MR2378039}. In what follows we will simply call these actions the direct product and the wreath product.

For a permutation group $G$ acting on the set $X$ we define the following numbers:
\begin{itemize}
	\item $f_K(G)$: the number of orbits of $G$ on the set of K-element subsets of $X$
	\item $F_K(G)$: the number of orbits of $G$ on the set of ordered K-tuples of distinct elements of $X$
	\item $F^{\star}_K(G)$: the number of orbits of $G$ on the set of all ordered K-tuples of elements of $X$
\end{itemize} 
By convention $f_0(G) =  F_0(G) = F^{\star}_0(G) = 1$. Moreover we have the elementary inequalities
\begin{equation} 
f_K \leq F_K \leq K!f_K \label{eq:approx} 
\end{equation}
It will be useful to consider infinite permutation groups. We call such a group $\emph{oligomorphic}$ if it has a finite number of orbits for all $K$, \ie $F_K < \infty$ for all $K$. 
An example is $S$, the permutation group of a countable set of elements $X$, for which $f_K=F_K=1$. We will also be interested in families of permutation groups $\{G_N\}_{N\in \mathbb{N}}$. We define:
\begin{definition}
	A family of permutation groups $G_N$ is called oligomorphic if the $F_K(G_N)$ converge pointwise, that is if
	\begin{equation}
	F_K(G_N)= F_K \quad N\;\textit{large enough}
	\end{equation}
\end{definition}
The point of this definition is that for an oligomorphic family the $N\to\infty$ limit of (\ref{Bantay}) is well-defined \cite{Belin:2014fna}.
An example is of course the family $S_N$, for which we have $F_K(S_N)\rightarrow F_K(S)$. Similar statements hold for the wreath product and the direct product action. In practice this means that we can compute the $F_K(S_N)$ from $F_K(S)$ as long as we choose $N$ much bigger than $K$.

\subsection{Cycle index}

\begin{definition}
Let $G_N$ be a permutation group on $N$ elements. For $\sigma \in G_N$, denote the number of cycles of length $k$, $1\leq k\leq N$ in the cycle decomposition of $\sigma$ by $m_k(\sigma)$. Then the cycle index of $G_N$ is the following polynomial in the variables $s_1,s_2,....,s_N$:
\begin{equation}
\chi_{G_N}(s_1,s_2,s_3,....,s_N)= \frac{1}{|G_N|}\bigg( \sum_{\sigma\in G_N} s_1^{m_1(\sigma)}s_2^{m_2(\sigma)}s_3^{m_3(\sigma)}\cdots s_N^{m_N(\sigma)}\bigg)
\end{equation} 
\end{definition}
The cycle indices of some groups are well known and can be found e.g in \cite{MR1311922} 
\begin{align}
&\mathrm{Cyclic \; Group \; }C_N: \quad \chi_{C_N}= \frac{1}{N}\sum_{d|N}\Phi(d)s_d^{N/d} \\ 
&\mathrm{Dihedral \; Group \;} D_N: \quad \chi_{D_N} = \chi_{C_N}\; + \; \begin{cases} s_1s_2^{(N-1)/2} \quad N\;\mathrm{ odd}, \\ (s_2^{N/2} + s_1^2s_2^{(N-2)/2})/2 \quad &N\;\mathrm{ even}\end{cases}\\
&\mathrm{Symmetric \; Group \;} S_N: \quad \chi_{S_N}= \frac{1}{N!}\sum_{\lambda \vdash N}\Bigg(\frac{N!}{\prod_{i=1}^N i^{m_i(\lambda)}m_i(\lambda)!}\Bigg)\cdot\prod_i s_i^{m_i(\lambda)} \\
&\mathrm{Alternating \; Group \;} A_N: \quad \chi_{A_N}= \frac{1}{N!}\sum_{\lambda \vdash N}\Bigg(\frac{N!\big(1 + (-1)^{(m_2(\lambda) + m_4(\lambda) + ... )}\big)}{\prod_{i=1}^N i^{m_i(\lambda)}m_i(\lambda)!}\Bigg)\cdot\prod_i s_i^{m_i(\lambda)}  
\end{align}
Moreover for the imprimitive and the direct product actions we have \cite{MR1066691}
\begin{align}
&\mathrm{Wreath \; Product \;} G\wr H: \quad\chi_{G \wr H} = \chi_H(\chi_G(s_1,s_2,s_3,...),\chi_G(s_2,s_4,s_6,...),...)\\
&\mathrm{Direct \; Product \;} G\times H: \quad\chi_{G \times H} = \chi_G \circ \chi_G \label{dPcycle}
\end{align}
where we define $s_i \circ s_j:=(s_{\lcm(i,j)})^{\gcd(i,j)}$ and extend it to monomials and polynomials.
Here $\Phi(d)$ is the Euler totient function.
We can express the generating functions of $f_K(G),F_K(G)$ and $F^\star_K(G)$ in terms of the cycle index of $G$\cite{MR1311922}:
\begin{align}\label{fgen}
f_{G_N}(t):=& \sum_{k=0}^Nf_kt^k=\chi_{G_N}(1+t,1+t^2,1+t^3,....,1+t^{N}) \\
F_{G_N}(t):=& \sum_{k=0}^N\frac{F_kt^k}{k!}=\chi_{G_N}(1+t,1,1,1,....,1)\label{Fgen} \\
F^\star_{G_N}(t):=& \sum_{k=0}^N\frac{F^\star_kt^k}{k!}=
\chi_{G_N}(e^t,1,1,1,....,1)  \label{Fstargen}
\end{align}
from which follows the identity
\be\label{Fstar}
F^\star_G(t)= F_G(e^t-1)\ .
\ee
From this one can derive the identity
\begin{equation} F^{\star}_K = \sum_{n=0}^KS_2(K,n)F_n\ . \label{eq:stirling} \end{equation}
Here $S_2(K,n)$ are the Stirling numbers of the second kind, and we say $F^{\star}_K$ is the Stirling transform of $F_K$. We can invert this using the inverse Stirling transform 
\begin{equation} 
F_K = \sum_{n=0}^KS_1(K,n)F^{\star}_n\ , \label{eq:invstirling} 
\end{equation}
where now the $S_1(K,n)$ are the Stirling numbers of the first kind.
The Stirling numbers of the second kind $S_2(n,k)$ are given by the number of ways of partitioning a set of $n$ elements into $k$ nonempty sets, so they are clearly non-negative integers. The Stirling numbers $S_1(n,k)$ of the first kind then are integers with sign $(-1)^{n-k}$.
Finally we have \cite{MR2378039}
\be\label{Fstardirect}
F_K^{\star}(G\times H)= F_K^{\star}(G)F_K^{\star}(H)\ .
\ee

\subsection{Wreath product}
Let us now estimate the growth of $F_K$ for the wreath product. For this we establish the following Lemma:
\begin{lemma}\label{Fdirect}
	Let $g(z):= e^z-1$. We then have 
	\begin{align}
	&a) \quad
	F_{S^d_\wr}(t)=
	\sum_{K=0}^\infty \frac{F_K^{S^d_{\wr}}}{K!}z^K=
	e^{g^{d-1}(z)} 
	\\
	&b) \quad \log(F^{S^{d}_\wr}_K) \simeq K\log(K)- K - K\log^{d}(K)
	\end{align}
\end{lemma}
Here we define $a_n \simeq b_n$ to mean that for $n\to\infty$, $a_n=b_n$ up to terms which grow slower than the slowest term written out explicitly in $a_n$ and $b_n$.
\begin{proof}
$a)$ 
This follows from the identity \cite{MR1311922}
\begin{align}
F_K(S \wr G) = F^{\ast}_K(G)\ ,
\end{align}
together with (\ref{Fstar}), which allows us to express the generating function recursively,
\be
F_{S\wr G}(t)= F^\star_G(t)= F_G(e^t-1)\ .
\ee
We have $F_K(S)=1$, so that for the $d$-fold wreath product $S^d_\wr$ we get indeed $a)$.
\\
$b)$ For the second part one needs the following theorem \cite{MR2172781}:
\begin{theorem}[Hayman]
	Let $f(z)=\sum a_Kz^K$ be an admissible function (for the case at hand this reduces to $f(z)$ being an entire function). Let $r_K$ be the positive real root of the equation $a(r_K)=K$, for each $K=1,2,...$ where:
	\begin{align}
	a(r)= r\cdot\frac{f'(r)}{f(r)} \\ 
	b(r)= r\cdot a'(r)
	\end{align}
	then 
	\begin{align}
	a_K\simeq \frac{f(r_K)}{r_K^K\sqrt{2\pi b(r_K)}} \label{Hayman1}
	\end{align}
\end{theorem}
Introduce a new variable $r =: \log^{d-1} x$. Defining $g(z):= e^z-1$, 
we have $a(r)=r g^{d-1}(r)' = r e^r e^{g^1(r)}\cdots e^{g^{d-2}(r)}$.
We are only interested in the behavior for large $K$, which means that we can assume that $r$ is large, and only keep the leading contribution. In particular we can approximate $g(r) \simeq e^r$, so that 
\bea
 a(x_K)&=&K \Leftrightarrow \\
 x_K \log(x_K)\cdot\log(\log(x_K))\cdot\log(\log(\log(x_K)))\cdots \log^{(n-1)^{\star}}(x_K)\simeq&& K \\ \nonumber \Leftrightarrow x_K\simeq \frac{K}{\log^{(n-1)^{\star}}(K)\cdot\log^{(n-2)^{\star}}(K)\cdots \log(K)}
\eea
giving
\be
r_K \simeq \log^{d-1} K\ .
\ee
Plugging this into (\ref{Hayman1}), we see that the leading contribution is given by just $r_K^K$,
\be
a_K \simeq r^{-K}_K \simeq  (\log^{d-1} K)^{-K}\ .
\ee
Taking into account the factor $K!$, we thus obtain
\be
\log(F^{S^{d}_\wr}_K) \simeq K\log(K)- K - K\log^{d}(K)
\ee

\end{proof}

\subsection{Direct product}
To compute the $F_K^{S^d_\times}$, we use the fact that $F^{S*}_K=B_K$, where $B_K$ are the Bell numbers, whose asymptotic behavior is given by \cite{MR671583}
\be
\log B_K \simeq K\log K - K\log\log K - K \ .
\ee
To get $F^*$ for the direct product, we can use (\ref{Fstardirect}) to find $F^{S^d_\times*}=B_K^d$. 
We can then immediately write 
\be \label{FK_Sd}
F_K^{S^d_\times}= \sum_{n=0}^KS_1(K,n)B_n^d \ .
\ee
(Alternatively we could have tried to use (\ref{dPcycle}) and (\ref{Fgen}).)
We now want to estimate the asymptotic behavior of this sum, which will lead us to 
\begin{prop}\label{conj:Fdirect}For $d\geq2$,
	$F_K^{S^d_\times}\approx e^{d K\log(K)}$ .
\end{prop}
\begin{proof}
To prove this we rewrite (\ref{FK_Sd}) using the series expansion of the Bell numbers 
\be
B_n= \frac{1}{e}\sum_{r\geq 0}\frac{r^n}{r!}\;, 
\ee
as well as the expression for the generating function of the Stirling numbers of the first kind \cite{MR2172781},
\bea
\quad 
\sum_{n=0}^K S_1(K,n)t^n= (t)_K := t(t-1)\cdots(t-K+1)\ ,
\eea
so that we have
\bea
F_K^{S^d_\times}=\sum_{n=0}^KS_1(K,n)B_n^d &= e^{-d}\sum_{n=0}^K\sum_{r_1,r_2,...r_d\geq 0}S_1(K,n)\frac{r_1^nr_2^n\cdots r_d^n}{r_1!r_2!\cdots r_d!}\\
&=e^{-d}\sum_{r_1,r_2,...r_d\geq 0}\sum_{n=0}^KS_1(K,n)\frac{(r_1r_2\cdots r_d)^n}{r_1!r_2!\cdots r_d!}\\
&=e^{-d}\sum_{r_1,r_2,...r_d\geq 0}\frac{(r_1r_2\cdots r_d)_K}{r_1!r_2!\cdots r_d!}
\label{FKalternate}
\eea
The theorem then follows immediately from Theorem 3 in \cite{Pittel2000}, which shows that (up to the order we are interested in) the right hand side of (\ref{FKalternate}) and $B_K^d$ grow at the same rate . 
\end{proof}

\section{Partition functions}
\subsection{Untwisted sector}\label{ss:ut}
The untwisted sector is given by the cycle index evaluated for arguments corresponding to the partition function,
\bea\label{Polya}
Z^u_{G_N}(\tau) &=& \chi(G_N; Z(\tau),\dots, Z(N\tau)) \cr
&=& {1\over |G_N|}\sum_{\bf j}A_{\bf j} Z(\tau)^{j_1}\dots Z(N\tau)^{j_N}\ .
\eea
Note that here we can take either $Z$ or $\tilde Z$. Since all the coefficients are positive, we can estimate a lower bound on the number of states by estimating
\be
\tilde Z(\tau) \geq 1 + \rho(\Delta/K)q^{\Delta/K}\ , \qquad Z(j\tau) \geq 1\ .
\ee
Here $\geq$ is understood to mean that all Fourier coefficients  of the two expressions satisfy the inequality.
Using (\ref{Fgen}) we then get a contribution to the term $q^{\Delta}$ of
\be
\frac{F_K}{K!} \rho(\Delta/K)^K\ .
\ee
We can now prove propositions~\ref{wreath} and (under the assumption that conjecture~\ref{conj:Fdirect} is true) proposition~\ref{direct} as straightforward corollaries.
\begin{proof}[Proof of Proposition~\ref{wreath}]
Assuming that $\Delta\gg K$, we can use the Cardy formula, $\rho(\Delta/K)\approx e^{\sqrt{4b\Delta/K}}$, where $b= \pi^2c/6$, with $c$ the central charge of the seed theory, giving
\be\label{Cardy}
\frac{F_K e^{\sqrt{4b K\Delta}}}{K!}\ .
\ee
Plugging in (\ref{direct}) we get 
\begin{align}
\frac{1}{K!}e^{\sqrt{4bK\Delta} + K\log(K)- K - K\log^{d}(K)}\approx e^{\sqrt{4bK\Delta}- K\log^{d}(K)}
\end{align}
for $K$- tuples. We want to choose $K$ in such a way to maximize the contribution to $\rho_{\infty}(\Delta)$. This gives
\begin{align}
\frac{\textit{d}}{\textit{d}K}\bigg|_{K=K^{\star}}\sqrt{4bK\Delta}- K\log^{d}(K)= 0 \Leftrightarrow K^{\star}\sim \frac{b\Delta}{\left(\log^{d}(b\Delta)\right)^2}
\end{align}
Choosing $\Delta\gg 1$, we see that the conditions for (\ref{Cardy}) are indeed satisfied.
We thus get a contribution to the growth of the form
\begin{align}
\rho_{S^d_\wr}(\Delta)\gtrsim e^{\frac{b\Delta}{\log^{d}(b\Delta)}}.
\end{align}
\end{proof}
\begin{proof}[Proof of Proposition~\ref{direct}]
Pick $\Delta_1$ to be the weight of the lightest state in the theory, so that $\rho(\Delta_1)\geq 1$. For $\Delta = K \Delta_1$ we thus have
\be
\rho_{S^d}(\Delta) \gtrsim \frac{F_K}{K!} \approx e^{(d-1)\frac{\Delta}{\Delta_1}{\log \Delta/\Delta_1}}\ .
\ee 
\end{proof}

\subsection{Twisted sector}
Let us now include the contribution of the twisted sectors, that is consider the entire expression (\ref{Bantay}). This is a much harder problem. In some special cases, we can find explicit expressions, such as in the symmetric case $G_N =S_N$: Here the orbifold partition function can be obtained from the cycle index in terms of Hecke operators $T_L$,
\be
Z_{S_N} = \chi_{S_N}(T_1 Z,T_2 Z,\ldots , T_N Z)\ .
\ee
(There is in fact a closed form expression for the generating function of the $Z_{S_N}$\cite{Dijkgraaf:1996xw}.)
The total expression is thus a polynomial in `modular invariant blocks', \ie terms of the form $T_LZ(\tau)$.
For general permutation orbifolds, we can try to mimic this behavior. For this purpose we introduce generalized blocks by defining the operators $\T_{\vec a}$
\be\label{genHecke}
\T_{\vec a}Z(\tau) = \sum_{\gamma\in \Gamma_{\vec p,\vec q}\backslash SL(2,\ZZ)} Z(a_1\gamma(\tau))\cdots Z(a_k\gamma(\tau))\ , \qquad a_i=:\frac{p_i}{q_i}\ , \ \gcd(p_i,q_i)=1\ .
\ee
Here $\Gamma_{\vec p,\vec q}$ is the common (right-)stabilizer of the matrices 
\be
\left(\begin{array}{cc}p_i&0\\0&q_i\end{array}\right) \in SL(2,\ZZ)\backslash M_{p_iq_i}
\ee 
for all $i$, where $M_m$ are the $2\times2$ integer matrices with determinant $m$. That is, it is the group that leaves  the
expression $Z(a_1\tau)\cdots Z(a_k\tau)$ invariant after accounting for the modular invariance of $Z(\tau)$. It is of finite index since there are only finitely many elements in $SL(2,\ZZ)\backslash M_{p_iq_i}$, so that (\ref{genHecke}) is a finite sum. In particular (\ref{genHecke}) can be written as a sum over products of terms of the form $Z(\tilde\gamma(\tau))$ with $\tilde\gamma \in SL(2,\ZZ)\backslash M_{p_iq_i}$. We can then find the usual representative for $\tilde \gamma$ to write this as $Z((a\tau+b)/c),\ ac=p_iq_i, 0\leq b < c$, which explains the connection to (\ref{Bantay}).

The motivation for introducing these operators is to express the orbifold partition functions as combinations of such operators, which are of course individually modular invariant. In a sense $\T$ generalize the ordinary Hecke operators $T$: For $a\in \mathbb{N}$ prime, $\T_a = a T_a$. Note that even for a single integer $a$, the two differ if $a$ has a prime factorisation with exponents bigger than one: We have for instance $\T_4 = 4 T_4 -T_1$.

For cyclic orbifolds we can obtain the full result from modular transformations of the untwisted sector partition function. This means that we can express it in terms of the cycle index and operators $\T_{\vec a}Z$ with integer $\vec a$. In particular we have
\begin{prop}
\be
Z_{C_N}=\frac{1}{N}\sum_{d|N}\Phi(d)\T_{\vec{L}_d}Z(\tau)\ ,
\qquad \vec{L}_d = (d)^{N/d}:=\underbrace{(d,\ldots, d)}_{N/d}\ .
\ee
\end{prop}
\begin{proof}
	The $Z_{(h,g)}$  in (\ref{Ztauchi}) transform under $SL(2,\ZZ)$ as $Z_{(h,g)}\mapsto Z_{(h^ag^b,h^cg^d)}$. To see this fix an orbit $\xi$. For the $S$ transformation, first note that $\tilde\lambda$, the size of the $h$ orbit in $\xi$ is given by $\tilde\lambda = \frac{\mu\lambda}{\gcd(\lambda,\kappa)}$, as follows from the fact that $h^\mu = g^\kappa$ on $\xi$. We then immediately have $\tilde\mu = |\xi|/\tilde\lambda=\gcd(\lambda,\kappa)$. Finally $\tilde \kappa$ is given by a solution of $\tilde\kappa \kappa + k\lambda = \gcd(\lambda,\kappa)$. A straightforward computation shows that indeed $Z((-\mu/\tau+\kappa)/\lambda) = Z((\tilde\mu\tau+\tilde\kappa)/\tilde\lambda)$. For the $T$ transformation simply note that $\kappa$ shifts by $\mu$, as $h\to hg$, which is of course compatible with the behavior for $\tau\to\tau+1$.

	Let $r$ be a generator for the cyclic group $C_N$. The sum then runs over $(r^m,r^n), n,m=1,\ldots,N$.	
	Let us start by considering the term $Z_{(r,1)}=Z(N\tau)$. There are $\Phi(N)$ such terms. Under $SL(2,\ZZ)$ transformations, we obtain exactly all the terms $Z_{(r^m,r^n)}$ with $\gcd(m,n,N)=1$. This can be seen by the following argument: If $\gcd(m,n)=1$, using B\'ezout's identity we can always find an $SL(2,\ZZ)$ matrix 
	with $a=m$ and $c=n$. If $\gcd(m,n)>1$, first transform to $Z_{(r^{\gcd(m,n)},1)}$ by the matrix with $a=\gcd(m,n)$ and $c=N$, which exists since $\gcd(m,n,N)=1$. Then transform with the matrix $a=m/\gcd(m,n), c=n/\gcd(m,n)$. Conversely, for any transformation we have $a = m + kN$ and $c= n + k'N$, so if $\gcd(m,n,N)>1$, then $\gcd(a,c)>1$, so that no corresponding matrix exists.
	
	Next for any $d|N$ consider the term $Z_{(r^d,1)}=Z(N\tau/d)^d$. There are $\Phi(N/d)$ such terms, and the orbit is given by all terms with $\gcd(m,n,N)=d$. To see this repeat the above argument after dividing everything by $d$. Summing over all $d|N$ accounts for all terms.
\end{proof}

Similarly, we find that more complicated permutation orbifolds can also be expressed in terms of such generalized Hecke operators. It is however no longer possible to write all the terms as modular transforms of the untwisted sector. This means that we will also have terms of the form $\T_{\vec a}Z$ where some of the $a_i$ are fractions rather than integers.

\subsection{Example: Orbifolds of the $E_8\times E_8\times E_8$ theory}
For illustration let us compute the orbifold for some subgroups of $S_{16}$. For concreteness we will orbifold the chiral theory given by the Niemeier lattice $E_8^3$, which has central charge 24 and partition function
\be
Z_{E_8^3}(\tau)= j(\tau) = q^{-1} + 744 + 196883 q +\ldots\ .
\ee
To give an impression of the contributions of the twisted sectors, we have computed the orbifold for the following six groups: 
\begin{align}
|S_{16}|=16!& & |S_4 \wr S_4| = (4!)^5
& & |S_2\wr S_2\wr S_2\wr S_2| = (2!)^{15}\\
|S_4 \times S_4| = (4!)^2&& |S_2\times S_2\times S_2 \times S_2| = (2!)^4
&&|\mathbf{1}| = 1
\end{align}
The actual computation of these orbifolds is straightforward, although in some cases quite tedious. Our strategy is to express (\ref{Bantay}) as a polynomial in blocks of the form $\T_{\vec a}Z$. For each monomial we compute the polar terms, which allows us to write the result as a polynomial in $j$, from which we can very efficiently extract as many non-polar coefficients as we want. Symmetric orbifolds can be efficiently computed using (\ref{genHecke}). Wreath products can be computed as iterated orbifolds, that is the permutation orbifold $G\wr H$ is given by orbifolding the $H$ permutation orbifold by $G$. Direct product orbifolds however we needed to evaluate from (\ref{Bantay}) directly. In the cases at hand we find
\be
Z_{S_2\times S_2\times S_2\times S_2}(\tau)=\frac{1}{2^4}\left(\T_{(1)^{16}}Z +15 \T_{(2)^{8}}Z + 210 \T_{(1)^4}Z \right)
\ee
and
\begin{align}
Z_{S_4\times S_4}(\tau)= \frac{1}{(4!)^2}\bigg[&\T_{(1^{16})}Z(\tau) + 12\T_{(1^8, 2^4)}Z(\tau) + 36\T_{(1^4, 2^6)}Z(\tau) + 51\T_{(2^8)}Z(\tau) +\\ \nonumber
&16\T_{(1^4, 3^4)}Z(\tau) + 64\T_{(1, 3^5)}Z(\tau) +
156\T_{(4^4)}Z(\tau) + 96\T_{(4, 12)}Z(\tau) + \\ \nonumber 
&96\T_{(1^2, 2, 3^2,6)}Z(\tau) + 48\T_{(2^2, 6^2)}Z(\tau) + 144\T_{(1^2, 2^4)}Z(\tau)+ 288\T_{(2, 4^2)}Z(\tau)+ \\ \nonumber 
&72\T_{(1^3, 2^2)}Z(\tau) +504\T_{(2^2)}Z(\tau)+ 750\T_{(1^4)}Z(\tau) + 864\T_{(1)}Z(\tau)+ 96\T_{(1, 3)}Z(\tau)+\\�\nonumber
&12\T_{\left(\left(\frac{1}{2}\right)^4, 2^4\right)}Z(\tau) + 72\T_{\left(\left(\frac{1}{2}\right)^4, 1, 2^2\right)}Z(\tau) + 72\T_{\left(\left(\frac{1}{2}\right)^2, 1^5, 2^2\right)}Z(\tau) +\\�\nonumber
&96\T_{\left(\left(\frac{1}{2}\right),(\left(\frac{3}{2}\right)^4), 2, 6\right)}Z(\tau) + 108\T_{\left(\left(\frac{1}{2}\right)^2, 1^2, 2^2\right)}Z(\tau) + 128\T_{\left(\left(\frac{1}{3}\right),1^2, 3\right)}\bigg]
\end{align}
We have plotted them in figure~\ref{f:rho}. For $S_{16}$ we can clearly see the behavior (\ref{symgrowth}) up to approximately $\Delta \sim 2c = 32$. At that point Cardy behavior starts, and all orbifolds converge to the same number of states.
For clarity we have also plotted $\rho(\Delta)$ relative to the symmetric orbifold density $\rho_{S_{16}}(\Delta)$ in figure~\ref{f:rhorel}. As expected from intuition, the order of the group is a good indicator for how powerful an orbifold is: In our examples, the bigger the order of the orbifold group, the fewer states the orbifolded theory has. This is of course obvious in the untwisted sectors, but our examples indicate that it also holds when one includes the twisted sectors.

\begin{figure}
	\begin{center}
\includegraphics[width=.7\textwidth]{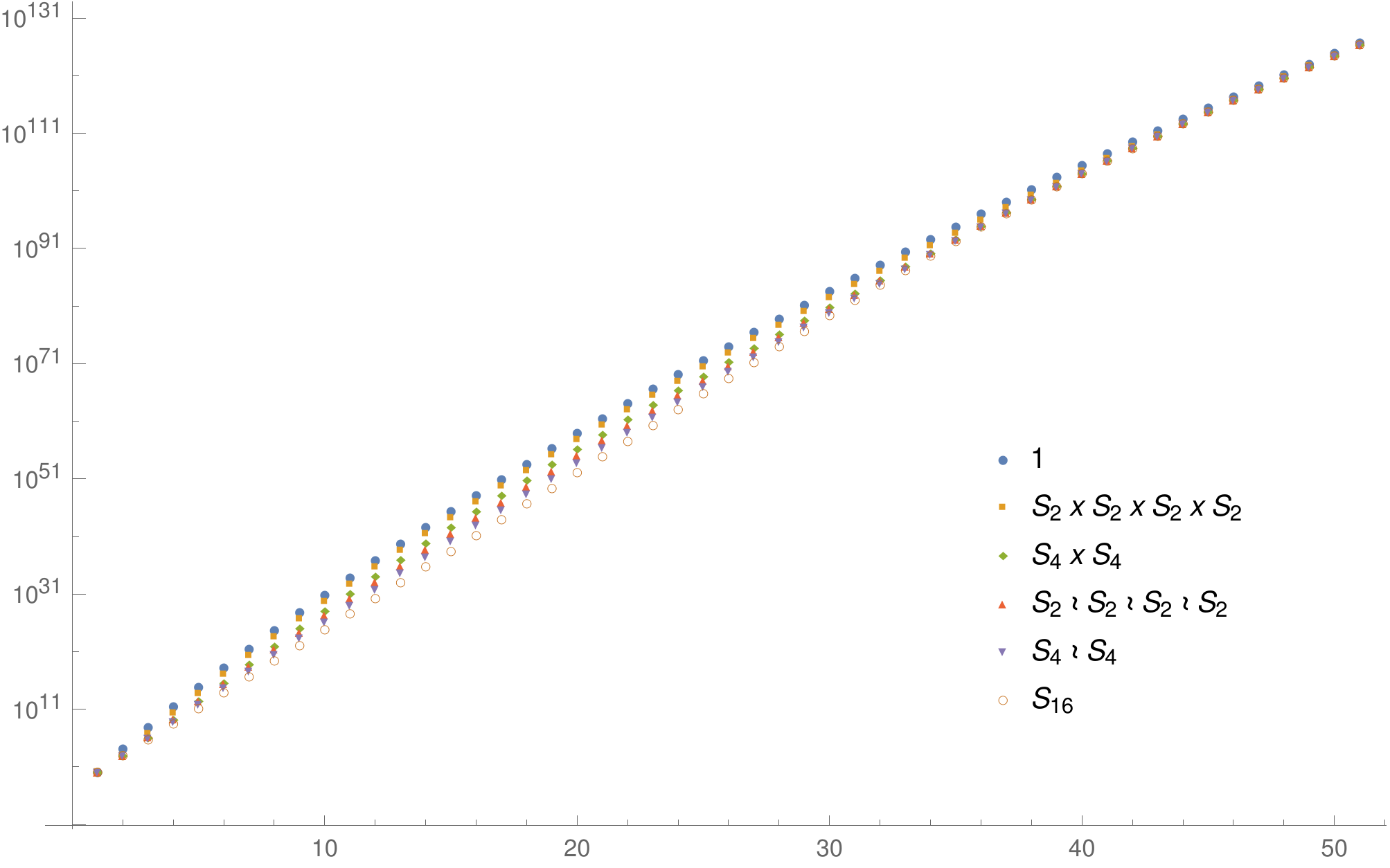}
	\end{center}
\caption{$\rho_G(\Delta)$ for various permutation orbifolds.}
\label{f:rho}
\end{figure}

\begin{figure}
	\begin{center}
	\includegraphics[width=.7\textwidth]{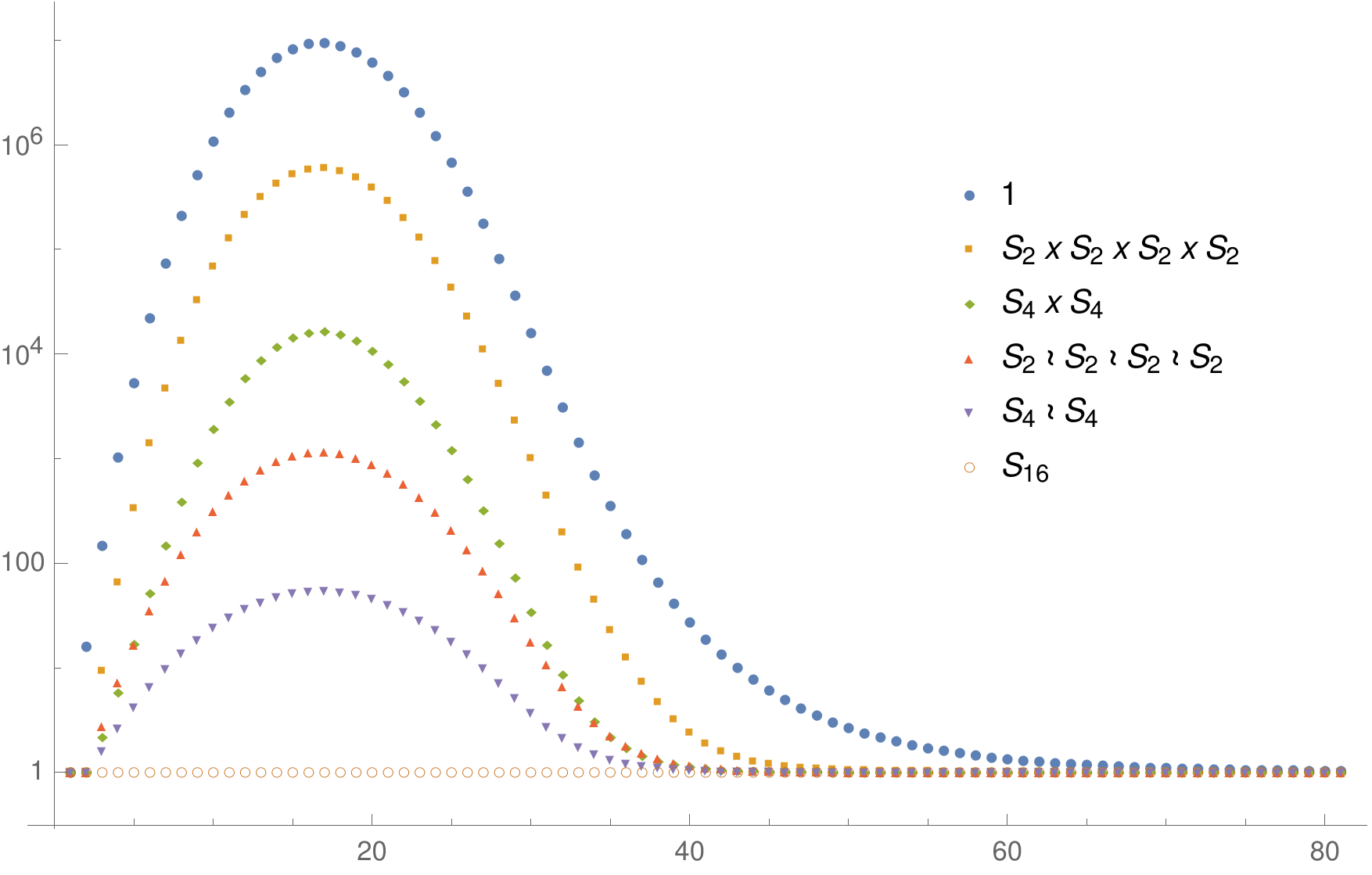}
	\end{center}
	\caption{$\rho_G(\Delta)/\rho_{S_{16}}(\Delta)$ for various permutation orbifolds.}
	\label{f:rhorel}
\end{figure}

\section*{Acknowledgments}
This work is partly based on the master thesis of one of us (BJM). CAK thanks the Harvard University High Energy Theory Group for hospitality. This work was performed in part at Aspen Center for Physics, which is supported by National Science Foundation grant PHY-1607611. CAK is supported by the Swiss National Science Foundation through the NCCR SwissMAP.

\bibliographystyle{utphys.bst}
\bibliography{ref}
\end{document}